\documentclass[12pt]{amsart}
\usepackage{amssymb,latexsym,amsthm, amsmath, comment,enumitem}
\usepackage[all]{xy}
\usepackage{graphicx, comment, eufrak, libertine}
\usepackage{fullpage}

\ifx\pdfoutput\undefined
\usepackage{hyperref}
\else
\usepackage[pdftex,colorlinks=true,linkcolor=blue,urlcolor=blue]{hyperref}
\fi

\theoremstyle{plain}
\newtheorem{theorem}{Theorem}

\newtheorem{proposition}{Proposition}

\newtheorem{lemma}{Lemma}
\theoremstyle{definition}
\newtheorem{definition}{Definition}
\newtheorem{remark}{Remark}

\newtheorem*{ack}{Acknowledgements}

\begin{document}

\title[Frenkel-Kontorova models on quasicrytals]
      {Equilibrium configurations for generalized Frenkel-Kontorova models on quasicrystals}

\author{Rodrigo Trevi\~no}
\address{Department of Mathematics\\
         University of Maryland, College Park}
\email{rodrigo@math.umd.edu}
\date{\today}

\begin{abstract}
I study classes of generalized Frenkel-Kontorova models whose potentials are given by almost-periodic functions which are closely related to aperiodic Delone sets of finite local complexity. Since such Delone sets serve as models for quasicrystals, this setup presents Frenkel-Kontorova models for the type of aperiodic crystals which have been discovered since Shechtman's discovery of quasicrystals. Here I consider models with configurations $u:\mathbb{Z}^r \rightarrow \mathbb{R}^d$, where $d$ is the dimension of the quasicrystal, for any $r$ and $d$. The almost-periodic functions used for potentials are called pattern-equivariant and I show that if the interactions of the model satisfies a mild $C^2$ requirement, and  if the potential satisfies a mild non-degeneracy assumption, then there exist equilibrium configurations of any prescribed rotation rotation number/vector/plane. The assumptions are general enough to satisfy the classical Frenkel-Kontorova models and its multidimensional analoges. The proof uses the idea of the anti-integrable limit.
\end{abstract}
\dedicatory{Dedicated to Rafael de la Llave on the occasion of his 60th birthday.}

\maketitle
\section{Introduction and statement of results}
The classical Frenkel-Kontorova model, first introduced to describe dislocations in solids \cite{FK}, aims to describe a chain of particles each interacting with its nearest neighbors in the presence of an external potential (see \cite{BK:FK} for other physical systems which are also described by this model). In the original model, the particles in the chain are indexed by $\mathbb{Z}$ through a function $u:\mathbb{Z}\rightarrow \mathbb{R}$, called a \textbf{configuration}. The system is given by the following formal sum, called the \textbf{action}:
\begin{equation}
\label{eqn:action1}
\mathcal{A}(u) = \sum_{i\in \mathbb{Z}}  \frac{1}{2} \left(\frac{du_i}{dt}\right)^2 + \frac{1}{2}( u_i-u_{i+1})^2 + (1-\cos(u_i)),
\end{equation}
where, as is customary, one writes $u_i := u(i)$. The term $\frac{1}{2}( u_i-u_{i+1})^2$ represents the nearest-neighbor interaction of the particles while the $1-\cos$ term represents the potential given by the (periodic) medium. Note that the sum is formal, so it does not need to converge a priori. The Euler-Lagrange equations for this action give
$$  \frac{\partial}{\partial u_i}\left(\sum_{i\in \mathbb{Z}}\frac{1}{2}( u_i-u_{i+1})^2 + (1-\cos(u_i)) \right) - \frac{d^2 u_i}{dt^2}= 0$$
for all $i\in\mathbb{Z}$. A particular type of solution to the variational problem above is the one where $\ddot{u}_i = 0$ for all $i$. This is called an \textbf{equilibrium configuration}. The condition for being an equilibrium configuration for (\ref{eqn:action1}) can be rewritten as
\begin{equation}
\label{eqn:cond0}
  u_{i+1} - 2u_i + u_{i-1} + \sin(u_i) = 0
 \end{equation}
for all $i\in\mathbb{Z}$. As such, although the action (\ref{eqn:action1}) may only formally be defined, the condition (\ref{eqn:cond0}) is well defined for all $i\in\mathbb{Z}$. It turns out that equilibrium configurations satisfying (\ref{eqn:cond0}) are in one to one correspondence with orbits of the \textbf{standard map} $f:\mathbb{R}^2\rightarrow \mathbb{R}^2$ by performing a change of coordinates which brings the Lagrangian formulation to a Hamiltonian one \cite{gole}. Since the potential is periodic then the map can be studied on the torus $\mathbb{R}^2/2\pi\mathbb{Z}^2$.

Let $u$ be an equilibrium configuration for (\ref{eqn:cond0}). The quantity
$$\lim_{i\rightarrow\infty} \frac{u_i}{i} = \rho(u),$$
if it exists, is called the \textbf{rotation number} of $u$. The question of whether there exist equilibrium configurations with prescribed rotation numbers is the same as a question of existence of invariant sets on the torus whereon the dynamics of the standard map are conjugate to a circle rotation with prescribed rotation number, which are usually defined by a Diophantine property. A rich theory has been developed to answer there types of questions (e.g. KAM \cite{llave:tutorial}, Aubry-Mather \cite{forni-mather}, etc). There are far too many relevant references to include that, in order not to offend anyone, I will not list any but refer the reader to the introduction of \cite{LlV:ground} which has a very thorough discussion.

There are several ways of generalizing the action given by $\mathcal{A}$ in (\ref{eqn:action1}). For example, one could take the potential to be any periodic function $V$, instead of the potential $(1-\cos)$, and still have a correspondence between equilibrium configurations and orbits of an area-preserving map of the torus. Another variation is changing the interaction between elements of a configuration: instead of nearest-neighbor interaction, the interactions can be more complex, even interactions with infinitely many elements. Or the interaction could be changed to be of the form $\frac{1}{2}( u_i-u_{i+1})^p$ for $p>2$, the equilibrium configurations of which are no longer in correspondence with orbits of a map.

One could also take the indices of the configurations to take values in a higher rank lattice: instead of looking for configurations of the form $u:\mathbb{Z}\rightarrow\mathbb{R}$, one can look for configurations of the form $u:\mathbb{Z}^r\rightarrow\mathbb{R}$ for some $r>1$. One such example is the multidimensional Frenkel-Kontorova model:
\begin{equation}
\label{eqn:actionHigh}
\mathcal{A}(u) = \sum_{i\in \mathbb{Z}^r}\frac{1}{2r}\sum_{|i-j| = 1}\| u_i-u_{j}\|^2 + (1-\cos(u_i)),
\end{equation}
the equilibrium configurations of which are solutions to the variational equation
$$(\triangle u)_i + \sin(u_i) = 0,$$
where $\triangle$ denotes the discrete Laplacian. In some cases, these configurations have a dynamical interpretation as orbits of maps in higher dimensional spaces. Again, interactions for higher rank lattices can be made more complex, and so can the potential. Finally, one could change the configurations to take values in a higher dimensional space, that is, configurations are functions $u:\mathbb{Z}^r\rightarrow \mathbb{R}^d$ for some $r>0$ and $d>0$. The general setup considered here is formulated using similar language to the one introduced in \cite{CandelLlave}.

Let $\mathcal{C}_{\mathbb{Z}^r,d}$ be the set of functions on $\mathbb{Z}^r$ with values in $\mathbb{R}^d$. Elements of this set are called \textbf{configurations}, where one usually writes $u_i := u(i)$.
\begin{definition}
  An \textbf{interaction} is a collection $\bar{\mathcal{S}} \subset \mathcal{P}(\mathbb{Z}^r)$ of finite subsets of $\mathbb{Z}^r$ along with real-valued maps $\{ H_B: B\in\bar{\mathcal{S}}\}$ that to each configuration $u$ associates a number $H_B(u)$ which depends only on the restriction of $u$ to $B$. Moreover, for a function $V:\mathbb{R}^d\rightarrow\mathbb{R}$, an \textbf{interaction with potential $V$} is an interaction $\{H_B\}$ with the aditional requirement that $\{i\}\in \bar{\mathcal{S}}$ for all $i\in\mathbb{Z}^r$ and if $B = \{i\}$ for some $i$ then $H_B(u) = V(u_i)$.
\end{definition}
For an interaction $\{H_B\}_{B\in \bar{\mathcal{S}}}$, let
$$\mathcal{S}  = \bar{\mathcal{S}} - \bigcup_{i\in\mathbb{Z}^r} \{i\}$$
which is the set of indices in $\mathcal{S}$ which have cardinality greater than one, and for any $i\in\mathbb{Z}^r$, denote by $\mathcal{S}_i\subset\mathcal{S}$ all the elements of $\mathcal{S}$ which contain $i$.
A \textbf{generalized Frenkel-Kontorova model} is given by an action of the form
$$\mathcal{A}(u) = \sum_{B\in\bar{\mathcal{S}}} H_B(u).$$
Moreover, if $\{H_B\}$ is an interaction with potential $V$, one has
\begin{equation}
  \label{eqn:action}
\mathcal{A}_V(u) = \sum_{B\in\mathcal{S}} H_B(u) + \sum_{i\in\mathbb{Z}^r}V(u_i),
\end{equation}
where where $V:\mathbb{R}^d\rightarrow\mathbb{R}$ is the potential. An \textbf{equilibrium configuration} for the generalized Frenkel-Kontorova model (\ref{eqn:action}) is a configuration $u$ satisfying the variational equation
\begin{equation}
\label{eqn:equilibrium}
\nabla_{u_i}\mathcal{A}(u) = \nabla_{u_i}\sum_{B\in\mathcal{S}} H_B(u) + \nabla V(u_i)  = 0
\end{equation}
for all $i\in\mathbb{Z}^r$. More specifically, since $H_B(u)$ is a function of $d|B|$ variables with values in $\mathbb{R}$, for any $i\in\mathbb{Z}^r$ and for any $B\subset\mathcal{S}$ containing $i$, $\nabla_{u_i}\sum_B H_B(u)$ is the vector in $\mathbb{R}^d$ obtained by taking $d$ partial derivatives with respect to the variables defined by the $d$ coordinates of $u_i$. Note that even if (\ref{eqn:action}) is just defined formally, (\ref{eqn:equilibrium}) can make sense under mild assumptions on the interaction $\{H_B\}$ (e.g. nearest-neighboor interaction).



The notion of rotation number can also be generalized for generalized Frenkel-Kontorova models as follows. Let $h:\mathbb{Z}^r\rightarrow \mathbb{R}^d$. A configuration $u$ is \textbf{of type $h$} if
\begin{equation}
\label{eqn:type}
\sup_{i\in\mathbb{Z}^r} \|u_i - h(i)\| < \infty.
\end{equation}
\begin{remark}
  \label{rem:type}
There are several classes of functions $h:\mathbb{Z}^r\rightarrow \mathbb{R}^d$ of increasing complexity which can be considered. For example, elements $h\in\mathrm{Hom}(\mathbb{Z}^r,\mathbb{R}^d)$ can be thought of as the natural generalization of rotation numbers. Another type of model configuration is something which could be called \textbf{heteroclinic}, and it is constructed as follows: given $ d,r\in\mathbb{N}$, let $P_{\sigma}$ be a subspace of $\mathbb{R}^r$ of dimension $\min\{d,r\}$ spanned by $d$ vectors $\sigma = \{v_1,\dots, v_d\}\in (\mathbb{R}^r)^d$ and let $M_\sigma$ be the matrix whose columns are the vectors in $\sigma$. Pick a basis $\{b_1,\dots, b_r\}$ of $\mathbb{Z}^r$ (where we consider the $b_i$ row vectors), pick two elements $\sigma^\pm_i\in (\mathbb{R}^r)^d$ for each $1\leq i \leq r$, and define $M^\pm_i:= M_{\sigma^\pm_i}$. Define $h:\mathbb{Z}^r\rightarrow \mathbb{R}^d$ as
  $$h\left(\sum_{i=1}^r n_i b_i\right) = \sum_{i=1}^r s(n_i)n_i b_i\cdot M_i^{s(n_i)},$$
  where $s(n_i)\in \{+,-\}$, depending on the sign of $n_i$, for all $\sum_{i=1}^r n_i b_i\in \mathbb{Z}^r$. As such, for any $j = 1,\dots, r$ and $g\in \mathbb{Z}^r$, we have that
  $$\sup_{n\in\mathbb{N}} \|h(g  \pm n b_j) \pm nb_j\cdot  M^\pm_j   \| < \infty.$$
  In other words, in the direction $nb_i\rightarrow \pm \infty$, $h$ behaves more and more like $h^\pm_i\in\mathrm{Hom}(\mathbb{Z}^r,\mathbb{R}^d)$ defined as $h^\pm_i(g) = g\cdot M^\pm_i$.

  When $r=d=1$, in the classical Frenkel-Kontorova model, since there is a bijection between equilibrium configurations and orbits of the standard map, there are types of ``chaotic'' configurations which can be modeled by a mixing subshift of the full 2-shift \cite{AA:AI}.
\end{remark}
  
  The set of configurations $u$ of type $h$ is denoted by $\mathcal{C}_{r,d}^h$. We denote by $\mathcal{L}(r,d)$ the set of functions $h:\mathbb{Z}^r\rightarrow \mathbb{R}^d$ which satisfy a Lipschitz condition. In such case the Lipschitz constant will be denoted by $L_h$. For $R\geq 0$, let
$$\mathcal{C}_{r,d}^{h,R}:= \{ u\in \mathcal{C}_{\mathbb{Z}^r,d}: \|u_i -h(i)\| \leq R \mbox{ for all }i\in \mathbb{Z}^r\}.$$
Then
$$ \mathcal{C}_{r,d}^h =  \bigcup_{R>0}\mathcal{C}_{r,d}^{h,R}.$$

In this paper I will show that equilibrium configurations exist of general types for generalized Frenkel-Kontorova models on quasicrystals. The word ``quasicrystal'' is generally taken to describe a type of crystal which exhibits a non-periodic structure \cite{lifshitz:What}. The first such material was discovered by Shechtman \cite{schechtman}. Although their structure is non-periodic, their atomic structure has long-range interactions and in many well-known cases they exhibit a great deal of symmetry and self-similarity.

The usual mathematical models for these materials are aperiodic Delone sets and/or aperiodic tilings. A \textbf{Delone set} $\Lambda\subset \mathbb{R}^d$ is a discrete subset of $\mathbb{R}^d$ satisfying two conditions:
\begin{description}
\item[Uniform discreteness]
There exists an $r_\Lambda>0$ such that any closed ball $B_{r_\Lambda}$ of radius $r_\Lambda$ centered anywhere in $\mathbb{R}^d$ intersects at most one point of $\Lambda$. This is also called the \textbf{packing radius};
\item[Relative denseness]
There exists a $R_\Lambda > 0$ such that any closed ball $B_{R_\Lambda}$ of radius $R_\Lambda$ centered anywhere in $\mathbb{R}^d$ intersects at least one point of $\Lambda$. This is also called the \textbf{covering radius}.
\end{description}
Delone sets and tilings offer views which are dual to each other. For example, the Delaunay triangulation defined by a Delone set defines a tiling with the same properties that the Delone set has. Conversely, given a tiling, the center of mass of each tile gives a Delone set with the same properties as the tiling.

Well known examples of the typed of tilings/sets which will be relevant here include:
\begin{itemize}
\item The Fibonacci and Thue-Morse tilings in one dimension,
\item the Penrose tiling and the Ammann-Beenker tilings in two dimensions,
  \item icosahedral tilings in three dimensions,
\end{itemize}
see \cite{BG:book1} for more details on these and many more examples. Two common properties for all the examples listed above are that they are repetitive and have finite local complexity (see \S \ref{sec:Delone} for definitions), which are reasonable properties to assume for models of quasicrystals. This paper will deal exclusively with Delone sets, but if the reader is more comfortable with tilings, then the reader can think of tilings.

The types of actions considered here are of the form (\ref{eqn:action}) with aperiodic potentials which depend locally on aperiodic point patterns and it will be shown that these models admit equilibrium configurations of very general prescribed types. For a Delone set $\Lambda$, the class of functions which will be considered for potentials are called $\Lambda$-equiviariant, or pattern-equivariant functions. In what follows, $\varphi_t(\Lambda)= \Lambda+t$ denotes the translation of $\Lambda$ by $t\in\mathbb{R}^d$.
\begin{definition}
A continuous function $f:\mathbb{R}^d\rightarrow\mathbb{R}$ is \textbf{$\Lambda$-equivariant} with range $R$ if there exists an $R>0$ such that if
$$\varphi_x(\Lambda)\cap B_R = \varphi_y(\Lambda)\cap B_R$$
then $f(x) = f(y)$. By the duality between tilings and Delone sets, there is an analogous notion for a tiling $\mathcal{T}$: a function $f$ is $\mathcal{T}$-equivariant if there is a $R>0$ such that the value of $f$ at $x$ only depends on the local configuration of the tiling $\mathcal{T}$ inside the ball of radius $R$ around $x$.
\end{definition}
In other words, the value of a $\Lambda$-equivariant function at any point $x\in\mathbb{R}^d$ depends only on the finite pattern of the set $\Lambda$ around the point $x$. An example of the types of potentials which are considered here and which is helpful to keep in mind is the following. Let $\Lambda\subset \mathbb{R}^d$ be a Delone set (e.g. the set of vertices of the Penrose tiling) and $\zeta:\mathbb{R}^d\rightarrow \mathbb{R}$ a smooth, radially symmetric bump function whose support is contained in a very small ball and of integral one. Then the function
$$V_\Lambda = \sum_{x\in\Lambda} \zeta * \delta_x$$
is a $\Lambda$-equivariant function. The function $V_\Lambda$ above is physically significant: by Dworkin's theorem \cite{dworkin}, the spectral measure associated with its (averaged) autocorrelation function is the diffraction measure of the material modeled by $\Lambda$.

As far as I know, there have been two works which have considered $\Lambda$-equivariant potentials for quasicrystals (that is, for repetitive, aperiodic $\Lambda$ of finite local complexity) for the Frenkel-Kontorova model are \cite{GGP:minimalFK, GPT:calibrated}. In \cite{GGP:minimalFK}, the authors show that, under some mild assumptions for the interactions, for a general class of Frenkel-Kontorova models with configurations $u:\mathbb{Z}\rightarrow\mathbb{R}$ and $\Lambda$-equivariant potentials, there exist minimal configurations (which are a type of equilibrium configurations) with any prescribed rotation number. In \cite{GPT:calibrated}, the authors aim to extend Aubry-Mather theory of the classical periodic model to the aperiodic setting. The authors show that under certain common hypothesis of the interaction (e.g. superlinearity, coerciveness, ferromagnetic) there exist calibrated configurations (which are a type of minimal configurations, and thus a type of equilibrium configurations) $u:\mathbb{Z}\rightarrow \mathbb{R}^d$ for Frenkel-Kontorova models with pattern equivariant potentials.

The purpose of this paper is to add to this short but growing literature on the subject extend these types of results for quasicrystals of dimension greater than one and configurations indexed by arbitrary lattices.
\begin{definition}
  \label{def:non-deg}
  A bounded $C^3$ function $V:\mathbb{R}^d\rightarrow \mathbb{R}$ is \textbf{non-degenerate} if there exists a point $z\in\mathbb{R}^d$ such that $\nabla V(z) = \bar{0}$ and whose Hessian at $z$ is invertible, i.e., if it has a non-degenerate critical point.
\end{definition}
The types of potentials which will be considered in this work have the following property:
\begin{enumerate}[label=\textbf{H\arabic*}]
  \setcounter{enumi}{0}
\item\label{H:non-deg} The potential $V$ has a non-degenerate critical point.
\end{enumerate}
The major assumption on the part of the interaction $\{H_B\}_{B\subset \mathcal{S}}$ which does not involve the potential will be
\begin{enumerate}[label=\textbf{H\arabic*}]
    \setcounter{enumi}{1}
\item\label{H:C2} for $h:\mathbb{Z}^r\rightarrow \mathbb{R}^d$ the interaction is $C^2$-\textbf{bounded} on $\mathcal{C}_{r,d}^{h,R}$: for all $R>0$ there exists a $B_{R,h,\{H_B\}}$ such that for any $B\in \mathcal{S}$ it will be assumed that $H_B$ is $C^2$ with
  \begin{equation}
    \label{eqn:C2}
     \sup_{i\in\mathbb{Z}^r} \sup_{u\in \mathcal{C}_{r,d}^{h,R}}\left\{  \left\| \nabla_{u_i} \left(\sum_{B\in \mathcal{S}_i} H_B\right)(u)\right\| +   \left\| \mathbf{H}\left(\sum_{B\in\mathcal{S}_i}H_B \right)(u) \right\|_{1,1} \right\}<B_{R,h,\{H_B\}},
  \end{equation}
where $\mathbf{H}(f)(u)$ is the Hessian of the function $f$ at $u$ and $\|\cdot\|_{1,1}$ denotes the $1,1$-norm $\|A\|_{1,1} = \sum_{i,j} |a_{i,j}|$ in the appropriate space of square matrices, which depends on $i$.
\end{enumerate}
Note that the condition \ref{H:C2} is one involving both the interaction $\{H_B\}$ as well as the type $h$. Some comments about hypothesis \ref{H:C2}, including how it is verified in a nearest-neighbor-type of interaction, can be found in \S \ref{subsec:H2} below following the statement of the second theorem.

The equilibrium configurations for (\ref{eqn:action}) will be coded by a set defined by a non-degenerate potential. More specifically, for a $C^r$, $\Lambda$-equivariant function $V$, $r>1$, let
$$\mathcal{Z}_V := \{x\in\mathbb{R}^d : \nabla V(x) = \bar{0}\mbox{ and }\mathrm{det}\,\mathbf{H}(V)(z)\neq 0\}.$$
This set will be used to code equilibrium configurations. That is, for
$$\mathcal{W}_{V} := \left\{ \mathbf{a}\in \mathcal{C}_{\mathbb{Z}^r, d} : \mathbf{a}_i\in \mathcal{Z}_V  \mbox{ for all }i\in\mathbb{Z}^r \right\},$$
and $\mathbf{a}\in \mathcal{W}_{V}$, there will be equilibrium configurations with values close to those of $\mathbf{a}$.
\begin{theorem}
\label{thm:main}
Let $\Lambda\subset \mathbb{R}^d$ be a repetitive Delone set of finite local complexity. Let $\{H_B\}_{\bar{\mathcal{S}}}$ be an interaction with potential $V$ for configurations $u:\mathbb{Z}^r\rightarrow\mathbb{R}^d$, where $V$ is $C^3$ and $\Lambda$-equivariant. Suppose $\{H_B\}_{B\in\bar{\mathcal{S}}}$ and $h:\mathbb{Z}^r\rightarrow \mathbb{R}^d$ satisfy \ref{H:non-deg} and \ref{H:C2}. Then there exists a $\lambda_0$ such that for any $\lambda>\lambda_0$ and for any $\mathbf{a}\in \mathcal{W}_V\cap \mathcal{C}_{\mathbb{Z}^r,d}^h$ there exists a unique equilibrium configuration $u$ of type $h$ for
\begin{equation}
  \label{eqn:main1}
  \mathcal{A}_V^\lambda(u) = \sum_{B\in\mathcal{S}} H_B(u) + \lambda \sum_{i\in\mathbb{Z}^r} V(u_i)
\end{equation}
which satisfies $\sup_i \|u_i - \mathbf{a}_i\| < \infty$.
\end{theorem}
Some of the best known Delone sets, such as those formed by the vertex set of the Penrose tiling, can be built from substitution rules, which give them the property of self-similarity. More generally, in the case of a self-affine Delone set $\Lambda$, one can use the self-affinity property to obtain equilibrium configurations.

More specifically, by definition (see \S \ref{sec:Delone}), if $\Lambda$ is self-affine, then there exists an expanding matrix $A\in GL^+(d,\mathbb{R})$ which represents the linear part of the self-affinity. In particular, if $V$ is a $\Lambda$-equivariant function, then so is the function $A^*V$. The function $A^* V$ is called a \textbf{scaling} of the function $V$. Indeed, for any $n\in\mathbb{Z}$, the function $(A^n)^*V$ is a scaling of $V$, where $(A^n)^*V$ is the pull-back of $V$ by the linear map $A^n$.

It is important to note the following subtlety: in general, if $\Lambda'$ is a Delone set, $B\in GL^+(d,\mathbb{R})$ an expanding matrix, and $f$ a $\Lambda'$-equivariant function, then $B^* f$ is not necessarily $\Lambda'$-equivariant. That is, rescaling any $\Lambda$-equivariant function by some expanding matrix will in general lose its relation to a Delone set $\Lambda$. What makes a self-affine Delone set $\Lambda$ special is that it comes with a matrix through which the scaling of a $\Lambda$-equivariant function is still $\Lambda$-equivariant. The examples mentioned above in one, two, and three dimensions are all self-similar, so they fall in this class.

Equilibrium configurations in the self-affine case will also be coded by sets, but they will be rescaled sets of zeros of non-degenerate potentials. More specifically, for $n\in\mathbb{Z}$, let $\mathcal{W}_V^n = \{\mathbf{a}\in\mathcal{C}_{\mathbb{Z}^r,d}: \mathbf{a}_i\in A^{-n}(\mathcal{Z}_V)\}$.
\begin{theorem}
\label{thm:main2}
Let $\Lambda\subset\mathbb{R}^d$ be a self-affine, repetitive Delone set of finite local complexity. Let $\{H_B\}_{\bar{\mathcal{S}}}$ be an interaction with potential $V$ for configurations $u:\mathbb{Z}^r\rightarrow\mathbb{R}^d$, where $V$ is $C^3$ and $\Lambda$-equivariant. Suppose $\{H_B\}_{B\in\bar{\mathcal{S}}}$ and $h:\mathbb{Z}^r\rightarrow \mathbb{R}^d$ satisfy \ref{H:non-deg} and \ref{H:C2}. Then there exists an $N$ such that for any $n>N$ and $\mathbf{a}\in \mathcal{W}_V^n\cap \mathcal{C}_{\mathbb{Z}^r,d}^h$ there exists a unique equilibrium configuration $u$ of type $h$ for
\begin{equation}
  \label{eqn:main2}
  \mathcal{A}^n_V(u) = \sum_{B\in\mathcal{S}} H_B(u) + \sum_{i\in\mathbb{Z}^r} (A^n)^*V(u_i)
  \end{equation}
which satisfies $\sup_i \|u_i - \mathbf{a}_i\| < \infty$, where $A$ is the self-affinity matrix associated with $\Lambda$.
\end{theorem}
The route to proving Theorems \ref{thm:main} and \ref{thm:main2} is using the ideas of the anti-integrable limit, first introduced in \cite{AA:AI} (see \cite{BT:anti} for an introduction). The strategy applied to the cases at hand is outlined in \S \ref{sec:anti}. Before that, some necessary background is covered in \S \ref{sec:Delone}.
\begin{remark}
  There is a recent construction \cite{ST:random, T:TTT} of (globally) random substitution tilings, which include self-affine tilings as particular cases, but are more general. Along with the construction comes a renormalization perspective which allows one to easily generalize Theorem \ref{thm:main2} to the case of random substituion tilings. The point is that instead of using powers $A^n$ of the same matrix in (\ref{eqn:main2}), one would use a particular product of $n$ matrices $A^{(n)} = A_n^\mathcal{T}\cdots A_1^\mathcal{T}$ which have to do with the random structure of the tiling $\mathcal{T}$, and the same result will hold. The random substitution tiling setup falls outside the scope of this paper but the enthusiastic reader will have not trouble seeing how Theorem \ref{thm:main2} also holds for tilings and Delone sets constructed in a (globally) random fashhion through the consctructions in \cite{ST:random, T:TTT}.
\end{remark}
\subsection{On the hypothesis \ref{H:C2}}
\label{subsec:H2}
  Let me revisit the hypothesis \ref{H:C2} and show how this hypothesis is verified for the most basic example, one with nearest-neighbor interactions as in (\ref{eqn:actionHigh}). More precisely, suppose for an interaction we have
  $$\sum_{B\in \mathcal{S}_i} H_B(u)= \frac{-1}{2r}\sum_{|i-j| = 1}\| u_i-u_{j}\|^2$$
  for all $i\in\mathbb{Z}^r$. Writing the vector $u_i = (u_i^1,u_i^2,\dots, u_i^d)$ in coordinates and setting $e_j\in \mathbb{Z}^k$ to be the $j^{th}$ standard basis element, we have
  \begin{equation}
    \label{eqn:verify}
    \nabla_{u_i}\left(\sum_{B\in \mathcal{S}_i} H_B \right)(u) = \frac{1}{r}\sum_{j=1}^r \left(\begin{array}{c} u^1_{i-e_j} -2u_i^1+u^1_{i+e_j} \\ \vdots \\   u^d_{i-e_j} -2u_i^d+u^d_{i+e_j} \end{array} \right).
    \end{equation}
  Now, if $u\in  \mathcal{C}^{h,R}_{r,d}$, then $u^\ell_{i-e_j} -2u_i^\ell+u^\ell_{i+e_j} = h^\ell(i-e_j) - 2 h^\ell(i) + h^\ell(i+e_j) + E^\ell_i(u,h)$ with $|E^\ell_i(u,h)|\leq 4R$ for all $\ell$. If, in addition, $h\in \mathcal{L}(r,d)$, (that is $h$ satisfies a Lipschitz condition), then we have that $|h^\ell(i-e_j) - 2 h^\ell(i) + h^\ell(i+e_j)|\leq 2L_h$ for all $\ell$. Putting these bounds together in (\ref{eqn:verify}), we have that
  $$\left\| \nabla_{u_i} \left(\sum_{B\in \mathcal{S}_i} H_B\right)(u)\right\| \leq \sqrt{d^2(2L_h + 4R)^2},$$
  which, I should point out, is independent of $i\in \mathbb{R}^r$. Now, since the interaction is a quadratic polynomial in the variables involved, the Hessian will be constant and independent of $i\in\mathbb{Z}^r$, which means that its $\|\cdot\|_{1,1}$ norm will also be bounded independently of $i\in \mathbb{Z}^r$ and $R$. This fact and the bound above shows that the interaction satisfies \ref{H:C2}. In general, if $h\in\mathcal{L}(r,d)$ and there is a $R$ such that the element of a configuration indexed by $i$ only interacts with particles indexed in $B_R(i)$ in a smooth way, then the condition \ref{H:C2} will be satisfied. Note that heteroclinic types of $h$ as described in Remark \ref{rem:type}, satisfy a Lipschitz condition.
\subsection{The $r$ question}Since the theorems above hold to a great deal of generality with respect to rank of the lattice $\mathbb{Z}^r$ and dimension of the quasicrystal $\Lambda$, one may ask whether, given a Delone set $\Lambda\subset \mathbb{R}^d$, there is a natural choice of $r$ such that the lattice $\mathbb{Z}^r$ is somehow connected to $\Lambda$. There are two approaches to this, one which seems much more natural to me than the other.

The first is choosing $r = d$. This is related to problems of bounded displacement or bi-Lipschitz equivalence of lattices to Delone sets, for which there is a vast literature (see \cite[\S 6]{ACCDP:survey} for background and references). However, there are examples of self-affine Delone sets $\Lambda\subset \mathbb{R}^d$ such that $\mathbb{Z}^d$ is not bounded-displacement equivalent to $\Lambda$. Moreover, whenever the two point sets are related in such a way, the relationship does not preserve the group structure of $\mathbb{Z}^d$. This does not look very natural to me, although there may be reasons for which one may choose this lattice which I do not know.

It seems to me that a more natural choice for $r$, given $\Lambda\subset\mathbb{R}^d$, is the rank of the Delone set $r_{\Lambda}$ as defined by Lagarias \cite{Lagarias:Geo1}. This is a natural choice for two reasons. The first is that the rank is finite for a class of Delone sets which includes all Delone sets of finite local complexity, so it is defined for all the Delone sets which are considered in this paper. If $\Lambda$ has finite local complexity then the abelian group
$$[\Lambda] = \mathbb{Z}[\mathbf{x}:\mathbf{x}\in\Lambda]$$
is finitely generated. Let $r_\Lambda$ be the rank of $[\Lambda]$. This is \textbf{the rank of $\Lambda$} and $[\Lambda]$ is isomorphic to $\mathbb{Z}^{r_\Lambda}$. By picking a basis $\{\mathbf{v}_1,\dots, \mathbf{v}_{r_\Lambda}\}$ of $[\Lambda]$ one defines the \textbf{address map} $\varphi:[\Lambda]\rightarrow \mathbb{Z}^{r_\Lambda}$ by
$$\varphi\left(\sum_{i}n_i\mathbf{v}_i\right) = (n_1,\dots, n_{r_\Lambda}).$$
Lagarias proved that the address map, which is defined for a broad class of Delone, satisfies a Lipschitz bound if and only if the Delone set has finite local complexity.

The rank of $\Lambda$ represents the number of ``internal dimensions'' of $\Lambda$, or the rank of $\Lambda$ as a quasilattice. More precisely, by picking generators of $[\Lambda]$ and defining an address map, the elements of $\Lambda$ can be labeled by elements of the subset $Y_\Lambda = \varphi(\Lambda)\subset \mathbb{Z}^{r_\Lambda}\subset \mathbb{R}^{r_\Lambda}$ for which there is a linear projection $\psi:\mathbb{R}^{r_\Lambda}\rightarrow \mathbb{R}^d$ which satisfies $\psi(Y_\Lambda) = \Lambda$. 

With all this in mind, it is possible to write a generalized Frenkel-Kontorova model for quasicrystals in the spirit of the original model, that is, one for which the non-potential interactions are of ``nearest-neighbor'' type. One can do this by picking $\bar{\mathcal{S}} \subset \mathcal{P}(Y_\Lambda)\subset \mathcal{P}(\mathbb{Z}^{r_\Lambda})$ for interactions $\{H_B\}_{B\in\bar{\mathcal{S}}}$ with $\Lambda$-equivariant potentials.

More specificaly, let $\Lambda$ be a repetitive Delone set $\Lambda\subset\mathbb{R}^d$ of finite local complexity and define $Y_\Lambda\subset \mathbb{Z}^{r_\Lambda}$ as above by picking an address map. For $\tau>1$ and $z\in Y_\Lambda$, let $N_\tau(z) = \overline{B_\tau(z)}\cap Y_\Lambda$ and define the action
$$\mathcal{A}^\tau_V(u) = \sum_{i\in Y_\Lambda}\frac{1}{2}\sum_{j\in N_\tau(i)}\|u_i - u_j\|^2 + \sum_{i\in Y_\Lambda} V(u_i).$$
By the Lipschitz property of the address map for Delone sets of finite local complexity, the interaction between particles is restricted to a finite neighborhood of each particle defined by $N_\tau$. By the discussion in \S \ref{subsec:H2} this type of interaction satisfies \ref{H:C2}. Thus, if $V$ satisfies \ref{H:non-deg}, by magnifying or scaling the potential, by the main theorems of this paper, equilibrium configurations exist.
\begin{ack}
I first learned about the Frenkel-Kontorova model from Rafael de la Llave. He also suggested that an approach using the anti-integrable limit may be fruitful. I am grateful to him for many conversations about the Frenkel-Kontorova, math, and life in general over many years. I am also grateful to two anonymous referees who have helped make the paper and its exposition much better.
\end{ack}
\section{Delone sets}
\label{sec:Delone}
Delone sets can be translated: for any $t\in\mathbb{R}^d$ the Delone set $\varphi_t(\Lambda) = \Lambda + t$ is the Delone set obtained by translating the Delone set $\Lambda$ by the vector $t\in\mathbb{R}^d$.

A finite subset $C\subset \Lambda$ of Delone set represents a \textbf{cluster}. A Delone set is \textbf{repetitive} if for any cluster $C\subset \Lambda$, there exists an $R>0$ such that any ball of radius $R$ centered anywhere in $\mathbb{R}^d$ contains a copy of $C$ in it.

A Delone set $\Lambda$ has \textbf{finite local complexity} if for any $R>0$ there are only finitely many clusters (up to translation) found in $B_R(y)\cap \Lambda$ for any $y\in\mathbb{R}^d$. All the Delone sets considered in this paper will be assumed to be repetitive and have finite local complexity.

One can impose a metric on the set of all translates of $\Lambda$:
$$d(\Lambda_1, \Lambda_2) = \min\left\{2^{-1/2}, \inf_{\varepsilon>0}\{B_{\varepsilon^{-1}} \cap \varphi_x(\Lambda_1)= B_{\varepsilon^{-1}} \cap \varphi_x(\Lambda_2)\mbox{ for some }x,y\in B_\varepsilon(\bar{0})  \} \right \}.$$
The \textbf{pattern space} of $\Lambda$ (sometimes also called the \emph{hull}), denoted by $\Omega_\Lambda$ is the closure of the set of all translations with respect to the metric above:
$$\Omega_\Lambda  = \overline{\{\varphi_t(\Lambda): t\in\mathbb{R}^d\}}.$$
Whenever $\Lambda$ has finite local complexity the pattern space $\Omega_\Lambda$ is a compact metric space. In most cases of interest, the translation action given by $\mathbb{R}^d$ on $\Omega_\Lambda$ is minimal and uniquely ergodic. Here I will only assume that the action is minimal, which implies repetitivity of the Delone set. 
\begin{definition}
A Delone set $\Lambda$ is \textbf{self-affine} if there exists a homeomorphism of the pattern space $\Phi_A:\Omega_\Lambda\rightarrow\Omega_\Lambda$ such that the conjugacy
\begin{equation}
  \label{eqn:conjugacy}
\Phi_A\circ\varphi_t = \varphi_{At}\circ\Phi_A
\end{equation}
holds for some expanding matrix $A\in GL^+(d,\mathbb{R})$, where $\Omega_\Lambda$ is the pattern space of $\Lambda$ and $\varphi_t$ denotes the translation action on the pattern space $\Omega_\Lambda$.
\end{definition}
The matrix $A$ is called the \textbf{self-affinity matrix} of $\Lambda$ if $\Lambda$ is self-affine and $A$ satisfies the conjugacy (\ref{eqn:conjugacy}). Its eigenvalues are listed by magnitude: $|\lambda_1|\geq |\lambda_2|\geq \cdots \geq |\lambda_d| \geq 1$.
\subsection{Pattern equivariant functions}
Let $\Lambda\subset \mathbb{R}^d$ be a Delone set of finite local complexity. For a finite dimensional vector space $V$ and Delone set $\Lambda$, the set $C^r_\Lambda(\mathbb{R}^d;V)$ denotes the space of $C^r$ $\Lambda$-equivariant functions with values in $V$ and will denote by $C^r_\Lambda$ the space of $\Lambda$-equivariant functions $f\in C^r(\mathbb{R}^d;\mathbb{R})$. For a self-affine Delone set $\Lambda$ with self-affinity matrix $A$, if $f$ is a continuous $\Lambda$-equivariant function, so is $A^* f$ (see \cite[\S 4]{ST:SS}). For $f\in C^3_\Lambda$ its Hessian matrix at the point $z$ is denoted by $\mathbf{H}(f)(z)$.

If $f$ is a $\Lambda$-equivariant function, then there exists a continuous function $h_f\in C(\Omega_\Lambda)$ such that
$$f(t) = h_f\circ \varphi_t(\Lambda).$$
If $\Lambda$ has finite local complexity, then $\Omega_\Lambda$ is compact, and so $h_f$ is bounded \cite[Proposition 22]{Kellendonk-Putnam:RS}.

\section{The anti-integrable limit}
\label{sec:anti}
This section briefly describes the idea behind the anti-integrable limit and introduces some notation which will be used later. The reader is refered to \cite{BT:anti} for a thorough introduction to the anti-integrable limit.

Let $\{H_B\}_{B\subset\bar{\mathcal{S}}}$ be an interaction with some pontential $V$ and recall that $\{H_B\}_{\bar{\mathcal{S}}}$ represents the part of the interaction which does not involve $V$. Denote by $n_i$ the number of other indices in the interaction $\mathcal{S}_i$ which are not $i$, that is, the number of other particles with which the particle with index $i\in\mathbb{Z}^r$ interacts. These indices will be denoted by $j_1,\dots, j_{n_i}$. As such, for an interaction $\{H_B\}_{B\subset \mathcal{S}}$ and any $i\in \mathbb{Z}^r$, let
\begin{equation}
  \label{eqn:Q}
\mathcal{Q}_{i} := \nabla_{u_i}  \sum_{B\subset\mathcal{S}_i} H_B:\mathbb{R}^{d(n_i+1)}\longrightarrow \mathbb{R}^d
\end{equation}
be the function that only depends on the $d(n_i+1)$ variables $u_{i}, u_{j_1},\dots, u_{j_{n_i}}\in \mathbb{R}^d$ obtained by taking $d$ partial derivatives with respect to the coordinate elements of $u_i$. As such, with a slight abuse on notation, if $u$ is a configuration then by $\mathcal{Q}_{i}(u)$ it is meant that $\mathcal{Q}_{i}$ only depends on the elements $u_{i}, u_{j_1},\dots, u_{j_{n_i}}$ of the configuration $u$.

For a $\Lambda$-equivariant function $V\in C_\Lambda^r$, $r>1$, let
$$\mathcal{Z}_V := \{x\in\mathbb{R}^d : \nabla V(x) = \bar{0}\mbox{ and }\mathrm{det}\,\mathbf{H}(V)(z)\neq 0\}.$$
Definition \ref{def:non-deg} requires non-degenerate functions to have only one non-degenerate point. However, if the function is $\Lambda$-equivariant for some repetitive $\Lambda$, then by repetitivity there exist infinitely many non-degenerate points.
\begin{lemma}
  \label{lem:repZero}
  Let $\Lambda$ be a repetitive Delone set and $f\in C^3_\Lambda$ be non-degenerate. Then the set $\mathcal{Z}_f$ is a relatively dense set.
\end{lemma}
\subsection{The anti-integrable limit for (\ref{eqn:main1})}
\label{subsec:anti1}
Recall that equilibrium configurations for $\mathcal{A}^\lambda_V$ are maps $u:\mathbb{Z}^r\rightarrow \mathbb{R}^d$ satisfying
$$\nabla_{u_i}\mathcal{A}^\lambda_V(u) = \mathcal{Q}_i(u) + \lambda\cdot \nabla V(u_i)  = 0,$$
for all $i\in \mathbb{Z}^r$. This is equivalent to
\begin{equation}
\label{eqn:condition}
\nabla V(u_i) = - \frac{\mathcal{Q}_i(u)}{\lambda}.
\end{equation}
Note that as $\lambda\rightarrow \infty$, the right side of (\ref{eqn:condition}) goes to zero. The idea of the anti-integrable limit is, first, noting that the ``equilibrium configuration'' for $\lambda = \infty$ consists of configurations taking values in the zero set of $\nabla V$. Thus, for $\lambda$ large enough, one should be able to find equilibrium configurations for $\mathcal{A}_V^\lambda$ by looking close to the zero set of $\nabla V$. Moreover, since the set $\mathcal{Z}_V$ is relatively dense (Lemma \ref{lem:repZero}), given any $h$ one can place any element $u_i$ of a configuration at a uniformly bounded distance from $h(i)$ in order to satisfy (\ref{eqn:type}).
\subsection{The anti-integrable limit for (\ref{eqn:main2})}
\label{subsec:anti2}
Let $\Lambda$ be a self-affine Delone set and let $A\in GL^+(d,\mathbb{R})$ be its affinity matrix. The idea to show that there exists an $n\in\mathbb{N}\cup\{0\}$ such that there exist equilibrium configurations for the action
$$ \mathcal{A}_V^n(u) = \sum_{B\in\mathcal{S}} H_B(u) + \sum_{i\in\mathbb{Z}^r} V(A^nu_i),$$
is very similar to the idea described in \S \ref{subsec:anti1}. In other words, equilibrium configurations for $\mathcal{A}^n_V$ are maps $u:\mathbb{Z}^r\rightarrow \mathbb{R}^d$ satisfying
$$\nabla_{u_i}\mathcal{A}^\lambda_V(u) = \mathcal{Q}_i(u) +  \nabla V(A^n u_i)\cdot A^n  = 0,$$
for all $i\in \mathbb{Z}^r$. This is equivalent to
\begin{equation}
\label{eqn:condition2}
\nabla V(A^nu_i) = - \mathcal{Q}_i(u)\cdot A^{-n}.
\end{equation}
As $n\rightarrow \infty$, the right side of (\ref{eqn:condition2}) goes to zero. Again, the idea of the anti-integrable limit is noting that the ``equilibrium configuration'' for $n = \infty$ consists of configurations taking values in ``the zero set of $(A^\infty)^*\nabla V$'', which is it not really well-defined. However, for $n$ large enough, one should be able to find equilibrium configurations for $\mathcal{A}_V^n$ by looking close to the zero set of $(A^n)^*\nabla V$. Again, since the set $\mathcal{Z}_V$ is relatively dense (Lemma \ref{lem:repZero}), given any $h$ one can place any element $u_i$ of a configuration at a uniformly bounded distance from $h(i)$ in order to satisfy (\ref{eqn:type}).
\section{Spaces of sequences and proof of Theorem \ref{thm:main}}
\label{sec:spaces}
Let
$$\mathcal{W}_{V} := \left\{ \mathbf{a}\in \mathcal{C}_{\mathbb{Z}^r, d} : \mathbf{a}_i\in \mathcal{Z}_V  \mbox{ for all }i\in\mathbb{Z}^r \right\}$$
and for any $h:\mathbb{Z}^r\rightarrow \mathbb{R}^d$, let $\mathcal{W}_{V}^h := \mathcal{W}_{V} \cap \mathcal{C}_{\mathbb{Z}^r, d}^h$. Finally, for $R\geq 0$, define $\mathcal{W}_{V}^{h,R} := \mathcal{W}_{V} \cap \mathcal{C}_{\mathbb{Z}^r, d}^{h,R}$. The sets $\mathcal{W}_{V}$ serve as sets which will ``code'' equilibrium configurations. In other words, given $\mathbf{a}\in\mathcal{W}_V$, it will be shown that, for $\lambda$ large enough, one can find equilibrium configurations $u\in \mathcal{C}_{\mathbb{Z}^r,d}$ such that $u_i$ and $\mathbf{a}_i$ are uniformly close. As such, the configuration $\mathbf{a}\in\mathcal{W}_V$ codes the equilibrium configuration $u$. Likewise, for any $h:\mathbb{Z}^r\rightarrow \mathbb{R}^d$, the sets $\mathcal{W}_V^h$ will code equilibrium configurations of type $h$.
\begin{lemma}
\label{lem:nonempty1}
  Let $\Lambda\subset\mathbb{R}^d$ be a repetitive Delone set, $r\in\mathbb{N}$, $V\in C^3_\Lambda$ a non-degenerate, $\Lambda$-equivariant function, and $h:\mathbb{Z}^r\rightarrow \mathbb{R}^d$. Then there exists a $R(V)$ such that for all $R>R(V)$, $\mathcal{W}_{V}^{h,R}\neq \varnothing$.
\end{lemma}
\begin{proof}
  By Lemma \ref{lem:repZero} there exists a $r_Z>0$ such that for any $x\in\mathbb{R}^d$, it follows that $B_{r_Z}(x)\cap \mathcal{Z}_V\neq \varnothing$. Thus, for any $i\in\mathbb{Z}^r$, there exists a $z_i\in B_{r_Z}(h(i))$ such that $\nabla V(z_i) = \bar{0}$ and $\|z_i - h(i)\| < 2r_Z$, so $R(V) = 2r_Z$.
\end{proof}
For $\eta>0$, recalling (\ref{eqn:Q}), define
\begin{equation}
\label{eqn:Weta}
  \mathcal{W}_{V}^\eta := \left\{ \mathbf{a}\in\mathcal{W}_{V} : \left\|  \mathcal{Q}_i(\mathbf{a})   \right\|  \leq \eta \mbox{ for all }i \in\mathbb{Z}^r\right\},
\end{equation}
$\mathcal{W}_{V}^{\eta,h}:= \mathcal{W}_{V}^\eta\cap\mathcal{W}_{V}^h$ and, for any $R>0$, $\mathcal{W}_{V}^{\eta,h,R}:= \mathcal{W}_{V}^\eta\cap\mathcal{W}_{V}^{h,R}.$ 
\begin{lemma}
  \label{lem:etaLarge}
Let $\Lambda\subset\mathbb{R}^d$ be a repetitive Delone set, $V\in C^3_\Lambda$ a non-degenerate, $\Lambda$-equivariant function, $h:\mathbb{Z}^r\rightarrow \mathbb{R}^d$, and $\{H_B\}$ an interaction satisfying \ref{H:C2}. Then there exists an $\eta_*>0$ and a $R_2>0$ such that for any $\eta>\eta_*$ and $R>R_2$, $\mathcal{W}_{V}^{\eta,h,R} \neq \varnothing$.
\end{lemma}
\begin{proof}
By Lemma \ref{lem:nonempty1} there exists a $R_*>0$ such that for all $R>R_*$, $\mathcal{W}_{V}^{h,R}\neq\varnothing$. Since the interaction $\{H_B\}$ satisfies \ref{H:C2}, by (\ref{eqn:C2}), the result follows with $\eta_* = B_{R,h, \{H_B\}}$ and $R_2 = R_*$.
\end{proof}

For $\eta, R$ large enough (Lemma \ref{lem:etaLarge}), $h:\mathbb{Z}^r\rightarrow\mathbb{R}^d$, $\mathbf{a}\in\mathcal{W}_{V}^{\eta,h,R}$, and $\varepsilon>0$, let
\begin{equation}
  \label{eqn:space}
  \Pi^{\varepsilon,h}_{V,\eta}(\mathbf{a}):= \left\{u\in\mathcal{C}_{\mathbb{Z}^r,d}: \|u_i - \mathbf{a}_i\|\leq \varepsilon \mbox{ for all }i \in\mathbb{Z}^r\right\}.
\end{equation}
Note that $\Pi^{\varepsilon,h}_{V,\eta}(\mathbf{a})\subset \mathcal{C}_{\mathbb{Z}^r,d}^h$. The set $\Pi^{\varepsilon,h}_{V,\eta}(\mathbf{a})$ is endowed with the metric
\begin{equation}
\label{eqn:metric}
\delta(u,u') = \sup_{i\in\mathbb{Z}^r}\|u_i - u'_i\|.
\end{equation}

For a non-degenerate potential $V\in C^3_\Lambda$, by the repetitivity and finite local complexity of $\Lambda$, there exists a $R_V>0$ such that for any $z\in \mathcal{Z}_V$, by the inverse function theorem, there is a local homeomorphism
\begin{equation}
  \label{eqn:locHomeo}
  \mathcal{K}_z : \overline{B_{R_V}(\bar{0})}\rightarrow U_z\subset \mathbb{R}^d
\end{equation}
which is a local inverse for $\nabla V$. That is, the map $\mathcal{K}_z$ satisfies
$\nabla V\circ \mathcal{K}_z(x) = x$ for any $x\in \overline{B_{R_V}(\bar{0})}$ and $\mathcal{K}_z(\bar{0}) = z$. Again, by finite local complexity, there are uniform bounds on the derivatives of the maps $\mathcal{K}_z$. In other words, writing $\mathcal{K}_z  = (\mathcal{K}_z^1,\dots, \mathcal{K}_z^d)$, there exists a $K_V$ such that for any $z\in \mathcal{Z}_V$ and $x\in \overline{B_{R_V}(\bar{0})}$,
\begin{equation}
  \label{eqn:DKbnd}
  \|\nabla \mathcal{K}_z^\ell(x)\| < K_V
\end{equation}
for all $\ell = 1,\dots, d$.
\begin{lemma}
  \label{lem:map}
  Under the hypotheses of Lemma \ref{lem:etaLarge}, there exists an $\lambda_*>0$ and $\varepsilon'>0$ such that for any $\lambda>\lambda_*$ and $\mathbf{a}\in\mathcal{W}_{V}^{\eta,h,R}$, the map $\Phi_{V,\mathbf{a},\eta}^{\varepsilon,h,\lambda}: \Pi^{\varepsilon,h}_{V,\eta}(\mathbf{a})\rightarrow \Pi^{\varepsilon',h}_{V,\eta}(\mathbf{a})$ given by
\begin{equation}
\label{eqn:map}
  (\Phi_{V,\mathbf{a},\eta}^{\varepsilon,h,\lambda}(u))_i = \mathcal{K}_{\mathbf{a}_i}\left( - \frac{\mathcal{Q}_i(u)}{\lambda}\right),
\end{equation}
 $i\in\mathbb{Z}^r$, is well-defined, where $\eta, R$ are chosen large enough as guaranteed by Lemma \ref{lem:etaLarge}.
\end{lemma}
\begin{proof}
  By the condition \ref{H:C2}, for $\lambda> B_{R,h,\{H_B\}}/R_V$, for $\mathbf{a}\in\mathcal{W}_V^{\eta,h,R}$ and $u\in \Pi_{V,\eta}^{\varepsilon,h}(\mathbf{a})$, we have that $R_i(\lambda):=\|\mathcal{Q}_i(u)\|/\lambda<R_V$ for any $i\in\mathbb{Z}^r$, so $-\mathcal{Q}_i(u)/\lambda$ is in the domain of the map $\mathcal{K}_{\mathbf{a}_i}$. Therefore the map is well defined. The value $\varepsilon'$ can be defined as
  $$\varepsilon' := \max_{z\in\mathcal{Z}_V} \max_{x\in \overline{B_{ R_i(\lambda)}(\bar{0})}}\| z - \mathcal{K}_z(x) \|$$
which, by finite local complexity and \ref{H:non-deg}, is well defined.
\end{proof}

\begin{proposition}
  \label{prop:contract}
    Under the hypotheses of Lemma \ref{lem:etaLarge}, there exists an $\lambda_*>0$ and $\varepsilon_*>0$ such that for any $\lambda>\lambda_*$, $\varepsilon<\varepsilon_*$ and $\mathbf{a}\in\mathcal{W}_{V}^{\eta,h,R}$, the map $\Phi_{V,\mathbf{a},\eta}^{\varepsilon,h,\lambda}$ defined by (\ref{eqn:map}) is a map from $\Pi^{\varepsilon,h}_{V,\eta}(\mathbf{a})$ into itself and is a contraction, where $\eta, R$ are chosen large enough as guaranteed by Lemma \ref{lem:etaLarge}.
\end{proposition}
\begin{proof}
Let $i\in \mathbb{Z}^r$. By expanding the map $(\Phi_{V,\mathbf{a},\eta}^{\varepsilon,h,\lambda}(u))_i$ in Taylor series, for two configurations $u,u'\in \Pi_{V,\eta}^{\varepsilon,h}(\mathbf{a})$ it follows that
\begin{equation}
\label{eqn:contract}
\begin{split}
  \|(\Phi_{V,\mathbf{a},\eta}^{\varepsilon,h,\lambda}(u))_i &- (\Phi_{V,\mathbf{a},\eta}^{\varepsilon,h,\lambda}(u'))_i\|^2 \\
  &\leq \|u_i-u'_i\|^2\max_\ell\frac{d^2\|\nabla \mathcal{K}_{\mathbf{a}_i}^\ell\|^2_1}{\lambda^2} \cdot \left\| \mathbf{H}\left(\sum_{B\in\mathcal{S}_i}H_B \right)(u) \right\|_{1,1}^2  + O(\|u_i-u'_i\|^4)\\
\mbox{(by (\ref{eqn:DKbnd}) and \ref{H:C2}) }  &\leq  \|u_i-u'_i\|^2\frac{K_V^2}{\lambda^2} B_{R,h, \{H_B\}}^2 + O(\|u_i-u'_i\|^4).
 \end{split}
\end{equation}
Since $\Lambda$ has finite local complexity, there are uniform bounds on the second order Taylor remainders of $\mathcal{K}_z$ for any $z\in \mathcal{Z}_V$. As such, the terms $O(\|u_i-u'_i\|^4)$ in (\ref{eqn:contract}) are uniformly bounded. Therefore, for any $\rho\in(0,1)$, there exist $\lambda$ large enough and $\varepsilon$ small enough such that
$$\|(\Phi_{V,n,\mathbf{a}}^{\varepsilon,h,\lambda}(u))_i - (\Phi_{V,n,\mathbf{a}}^{\varepsilon,h,\lambda}(u'))_i\| \leq \rho \|u_i-u'_i\|,$$
which shows $\Phi_{V,\mathbf{a}, \eta}^{\varepsilon,h,\lambda}$ is a contraction.
\end{proof}
\begin{proof}[Proof of Theorem \ref{thm:main}]
Let $\Lambda\subset \mathbb{R}^d$ be a repetitive Delone set of finite local complexity and let $V\in C^3_\Lambda$ satisfy \ref{H:non-deg}. For $h:\mathbb{Z}^r\rightarrow \mathbb{R}^d$ and an interaction $\{H_B\}$ satisfying \ref{H:C2}, by Lemma \ref{lem:etaLarge}, there exist $\eta, R>0$ such that $\mathcal{W}_V^{\eta,h,R}\neq \varnothing$.

By Proposition \ref{prop:contract}, there exist $\lambda$ and $\varepsilon$ such that the map $\Phi_{V,\mathbf{a},\eta}^{\varepsilon,h,R}$ from $\Pi_{V,\eta}^{\varepsilon, h,R}(\mathbf{a})$ to itself, defined by (\ref{eqn:map}), is a contraction for any $\mathbf{a}\in\mathcal{W}_V^{\eta,h,R}$. Thus, for $\mathbf{a}\in\mathcal{W}_V^{\eta,h,R}$ the map $\Phi_{V,\mathbf{a},\eta}^{\varepsilon,h,R}$ has a unique fixed point: there exists a unique configuration $u\in \Pi_{V,\eta}^{\varepsilon, h,R}(\mathbf{a})$ such that 
\begin{equation}
  \label{eqn:fixedPt}
  (\Phi_{V,\mathbf{a},\eta}^{\varepsilon,h,\lambda}(u))_i = u_i = \mathcal{K}_{\mathbf{a}_i}\left( - \frac{\mathcal{Q}_i(u)}{\lambda}\right).
\end{equation}
Taking $\nabla V$ of both sides of (\ref{eqn:fixedPt}), one obtains (\ref{eqn:condition}), which says $u$ is an equilibrium configuration. Moreover, since $\mathbf{a}\in \mathcal{C}_{\mathbb{Z}^r,d}^h$, by (\ref{eqn:space}), $u\in \mathcal{C}_{\mathbb{Z}^r,d}^h$.
\end{proof}
\section{Proof of Theorem \ref{thm:main2}}
The proof of Theorem \ref{thm:main2} follows the same path as the proof of Theorem \ref{thm:main} which is given in \S \ref{sec:spaces}. This section goes over the proof of Theorem \ref{thm:main2} but will reference \S \ref{sec:spaces} frequently as the proof is similar to that of Theorem \ref{thm:main}.

This section considers self-affine, repetitive Delone sets $\Lambda\subset\mathbb{R}^d$ of finite local complexity. Let $A\in GL^+(d,\mathbb{R})$ be the expanding affinity matrix associated with $\Lambda$ with eigenvalues $\lambda_1\geq\cdots\geq \lambda_d>1$. Let
$$\mathcal{Z}_V^n := A^{-n}(\mathcal{Z}_V).$$
Since $A^{-n}$ is a contraction for every $n>0$ and, by Lemma \ref{lem:repZero}, the set $\mathcal{Z}_V$ is relatively dense, the following is straight-forward.
\begin{lemma}
  \label{lem:repZero2}
  Let $\Lambda$ be a repetitive Delone set and $f\in C^3_\Lambda$ be non-degenerate. Then for any $n\in \mathbb{Z}$ the set $\mathcal{Z}_f^n$ is a relatively dense set. Moreover, if $m>n$, then the covering radius of $\mathcal{Z}_f^n$ is larger than the covering radius of $\mathcal{Z}_f^m$.
\end{lemma}
Let
$$\mathcal{W}_{V}^n := \left\{ \mathbf{a}\in \mathcal{C}_{\mathbb{Z}^r, d} : \mathbf{a}_i\in \mathcal{Z}_V^n  \mbox{ for all }i\in\mathbb{Z}^r \right\}$$
and for any $h:\mathbb{Z}^r\rightarrow \mathbb{R}^d$, let $\mathcal{W}_{V}^{h,n} := \mathcal{W}_{V}^n \cap \mathcal{C}_{\mathbb{Z}^r, d}^h$. Finally, for $R\geq 0$, define $\mathcal{W}_{V}^{h,R,n} := \mathcal{W}_{V}^n \cap \mathcal{C}_{\mathbb{Z}^r, d}^{h,R}$. As in the previous section, these sets will be used to code equilibium configurations of the right type. The proof of the following lemma is the same as the proof of Lemma \ref{lem:nonempty1} with the aditional consideration of Lemma \ref{lem:repZero2} since the set $\mathcal{Z}_V^n$ is relatively dense for all $n\in\mathbb{N}$ with decreasing packing radius as $n$ increases.
\begin{lemma}
\label{lem:nonempty2}
  Let $\Lambda\subset\mathbb{R}^d$ be a repetitive Delone set, $r\in\mathbb{N}$, $V\in C^3_\Lambda$ a non-degenerate, $\Lambda$-equivariant function, and $h:\mathbb{Z}^r\rightarrow \mathbb{R}^d$. Then there exists a $R(V)$ such that for all $R>R(V)$, $\mathcal{W}_{V}^{h,R,n}\neq \varnothing$ for all $n\in\mathbb{N}\cup\{0\}$.
\end{lemma}
For an interaction $\{H_B\}_{B\subset \mathcal{S}}$, $\eta>0$, and $n\in\mathbb{N}\cup\{0\}$, let $\mathcal{W}_{V}^{\eta,h,n}:= \mathcal{W}_{V}^\eta\cap\mathcal{W}_{V}^{h,n}$ and, for any $R>0$, $\mathcal{W}_{V}^{\eta,h,R,n}:= \mathcal{W}_{V}^\eta\cap\mathcal{W}_{V}^{h,R,n}$ where $\mathcal{W}_{V}^\eta$ is defined in (\ref{eqn:Weta}). The proof of the following lemma follows the same liness as the proof of Lemma \ref{lem:etaLarge} in the previous section.
\begin{lemma}
  \label{lem:etaLarge2}
Let $\Lambda\subset\mathbb{R}^d$ be a repetitive Delone set, $V\in C^3_\Lambda$ a non-degenerate, $\Lambda$-equivariant function, $h:\mathbb{Z}^r\rightarrow \mathbb{R}^d$, and $\{H_B\}$ an interaction satisfying \ref{H:C2}. Then there exists an $\eta_*>0$ and a $R_2>0$ such that for any $\eta>\eta_*$ and $R>R_2$, $\mathcal{W}_{V}^{\eta,h,R,n} \neq \varnothing$ for any $n\in\mathbb{N}\cup\{0\}$.
\end{lemma}

For $\eta, R$ large enough (Lemma \ref{lem:etaLarge2}), $h:\mathbb{Z}^r\rightarrow \mathbb{R}^d$, $n\in\mathbb{N}\cup\{0\}$, $\mathbf{a}\in\mathcal{W}_{V}^{\eta,h,R,n}$, and $\varepsilon>0$, let
\begin{equation}
  \label{eqn:space2}
  \Pi^{\varepsilon,h}_{V,\eta}(\mathbf{a}):= \left\{u\in\mathcal{C}_{\mathbb{Z}^r,d}: \|u_i - \mathbf{a}_i\|\leq \varepsilon \mbox{ for all }i \in\mathbb{Z}^r\right\},
\end{equation}
which is endowed with the metric (\ref{eqn:metric}). Note that $\Pi^{\varepsilon,h}_{V,\eta}(\mathbf{a})\subset \mathcal{C}_{\mathbb{Z}^r,d}^h$. Moreover, if $\mathbf{a}\in \mathcal{W}_V^n$, then $A^n \mathbf{a}_i\in\mathcal{Z}_V$ for all $i\in\mathbb{Z}^r$ and so the map $\mathcal{K}_{A^n\mathbf{a}_i}$ from (\ref{eqn:locHomeo}) is defined.

\begin{lemma}
  \label{lem:map2}
  Under the hypotheses of Lemma \ref{lem:etaLarge2}, there exists an $N\geq 0$ and $\varepsilon'>0$ such that for any $n\geq N$ the map $\Phi_{V,\mathbf{a},\eta}^{\varepsilon,h,n}: \Pi^{\varepsilon,h}_{V,\eta}(\mathbf{a})\rightarrow \Pi^{\varepsilon',h}_{V,\eta}(\mathbf{a})$ given by
\begin{equation}
\label{eqn:map2}
  (\Phi_{V,\mathbf{a},\eta}^{\varepsilon,h,n}(u))_i = A^{-n}\cdot \mathcal{K}_{A^n\mathbf{a}_i}\left( - \mathcal{Q}_i(u)\cdot A^{-n}\right),
\end{equation}
 $i\in\mathbb{Z}^r$, is well-defined for any $\mathbf{a}\in\mathcal{W}_{V}^{\eta,h,R,n}$, where $\eta, R$ are chosen large enough as guaranteed by Lemma \ref{lem:etaLarge2}.
\end{lemma}
\begin{proof}
  By \ref{H:C2}, for $n> \log(B_{R,h,\{H_B\}}/R_V)/\log(\lambda_d)$, $\mathbf{a}\in\mathcal{W}_V^{\eta,h,R,n}$ and $u\in \Pi_{V,\eta}^{\varepsilon,h}(\mathbf{a})$, the bound $R_i(n):= \|\mathcal{Q}_i(u)\cdot A^{-n}\|<R_V$ holds for any $i\in\mathbb{Z}^r$, so $-\mathcal{Q}_i(u)\cdot A^{-n}$ is in the domain of the map $\mathcal{K}_{A^n \mathbf{a}_i}$. As such, the map is well defined. The value $\varepsilon'$ can be defined as
  $$\varepsilon' := \max_{z\in\mathcal{Z}_V} \max_{x\in \overline{B_{R_i(n)}(\bar{0})}}\frac{\| z - \mathcal{K}_z(x) \|}{\lambda_d^n}$$
which, by finite local complexity and \ref{H:non-deg}, is well defined.
\end{proof}

\begin{proposition}
  \label{prop:contract2}
    Under the hypotheses of Lemma \ref{lem:etaLarge2}, there exists an $N_*>0$ and $\varepsilon_*>0$ such that for any $n>N_*$ and $\varepsilon<\varepsilon_*$, the map $\Phi_{V,\mathbf{a},\eta}^{\varepsilon,h,n}$ defined by (\ref{eqn:map2}) is a map from $\Pi^{\varepsilon,h}_{V,\eta}(\mathbf{a})$ into itself and is a contraction for any $\mathbf{a}\in\mathcal{W}_{V}^{\eta,h,R,n}$, where $\eta, R$ are chosen large enough as guaranteed by Lemma \ref{lem:etaLarge2}.
\end{proposition}
\begin{proof}
Let $i\in \mathbb{Z}^r$. By expanding the map $(\Phi_{V,\mathbf{a},\eta}^{\varepsilon,h,n}(u))_i$ in Taylor series, for two configurations $u,u'\in \Pi_{V,\eta}^{\varepsilon,h,n}(\mathbf{a})$ it follows that
\begin{equation}
\label{eqn:contract2}
\begin{split}
  \|(\Phi_{V,\mathbf{a},\eta}^{\varepsilon,h,n}(u))_i &- (\Phi_{V,\mathbf{a},\eta}^{\varepsilon,h,n}(u'))_i\|^2\\
  &\leq \|u_i-u'_i\|^2\max_\ell\frac{d^2\|\nabla \mathcal{K}_{A^n \mathbf{a}_i}^\ell\|^2_1}{\lambda_d^{2n}} \cdot 
  \left\| \mathbf{H}\left(\sum_{B\in\mathcal{S}_i}H_B \right)(u)  \right\|_{1,1}^2 \lambda^{-2n}_d + O(\|u_i-u'_i\|^4)\\
  &\leq  \|u_i-u'_i\|^2 K_V^2 B_{R,h, \{H_B\}}^2\lambda_d^{-4n} + O(\|u_i-u'_i\|^4) \;\;\;\;\;\;\;\;\mbox{(by (\ref{eqn:DKbnd}) and \ref{H:C2}) }.
 \end{split}
\end{equation}
Since $\Lambda$ has finite local complexity, there are have uniform bounds on the second order Taylor remainders of $\mathcal{K}_z$ for any $z\in \mathcal{Z}_V$. As such, the terms $O(\|u_i-u'_i\|^2)$ in (\ref{eqn:contract2}) are uniformly bounded. Therefore, for any $\rho\in(0,1)$, there exist $n$ large enough and $\varepsilon$ small enough such that
$$\|(\Phi_{V,n,\mathbf{a}}^{\varepsilon,h,n}(u))_i - (\Phi_{V,n,\mathbf{a}}^{\varepsilon,h,n}(u'))_i\| \leq \rho \|u_i-u'_i\|,$$
which shows $\Phi_{V,\mathbf{a}, \eta}^{\varepsilon,h,n}$ is a contraction.
\end{proof}
\begin{proof}[Proof of Theorem \ref{thm:main2}]
Let $\Lambda\subset \mathbb{R}^d$ be a self-affine, repetitive Delone set of finite local complexity and let $V\in C^3_\Lambda$ satisfy \ref{H:non-deg}. For $h:\mathbb{Z}^r\rightarrow \mathbb{R}^d$ and an interaction $\{H_B\}$ satisfying \ref{H:C2}, by Lemma \ref{lem:etaLarge}, there exist $\eta, R>0$ such that $\mathcal{W}_V^{\eta,h,R}\neq \varnothing$. Let $\mathbf{a}\in\mathcal{W}_V^{\eta,h,R}$.

By Proposition \ref{prop:contract2}, there exist $n\geq 0$ and $\varepsilon>0$ such that the map $\Phi_{V,\mathbf{a},\eta}^{\varepsilon,h,R,n}$ from $\Pi_{V,\eta}^{\varepsilon, h,R}(\mathbf{a})$ to itself, defined by (\ref{eqn:map2}), is a contraction. Thus, this map has a unique fixed point: there exists a unique configuration $u\in \Pi_{V,\eta}^{\varepsilon, h,R}(\mathbf{a})$ such that 
\begin{equation}
  \label{eqn:fixedPt2}
  (\Phi_{V,\mathbf{a},\eta}^{\varepsilon,h,n}(u))_i = u_i = A^{-n}\cdot \mathcal{K}_{A^n \mathbf{a}_i}\left( - \mathcal{Q}_i(u)\cdot A^{-n}\right).
\end{equation}
Multiplying on the left by $A^n$ and then taking $\nabla V$ of both sides of (\ref{eqn:fixedPt2}), one obtain (\ref{eqn:condition2}), which says $u$ is an equilibrium configuration. Moreover, since $\mathbf{a}\in \mathcal{C}_{\mathbb{Z}^r,d}^h$, by (\ref{eqn:space}), $u\in \mathcal{C}_{\mathbb{Z}^r,d}^h$.
\end{proof}

\bibliographystyle{amsalpha}
\bibliography{biblio}

\newcommand{\etalchar}[1]{$^{#1}$}
\providecommand{\bysame}{\leavevmode\hbox to3em{\hrulefill}\thinspace}
\providecommand{\MR}{\relax\ifhmode\unskip\space\fi MR }
\providecommand{\MRhref}[2]{%
  \href{http://www.ams.org/mathscinet-getitem?mr=#1}{#2}
}
\providecommand{\href}[2]{#2}
\begin{thebibliography}{APCC{\etalchar{+}}15}

\bibitem[AA90]{AA:AI}
Serge Aubry and Gilles Abramovici, \emph{Chaotic trajectories in the standard
  map. {T}he concept of anti-integrability}, Phys. D \textbf{43} (1990),
  no.~2-3, 199--219. \MR{1067910}

\bibitem[APCC{\etalchar{+}}15]{ACCDP:survey}
Jos\'e Aliste-Prieto, Daniel Coronel, Mar\'\i a~Isabel Cortez, Fabien Durand,
  and Samuel Petite, \emph{Linearly repetitive {D}elone sets}, Mathematics of
  aperiodic order, Progr. Math., vol. 309, Birkh\"auser/Springer, Basel, 2015,
  pp.~195--222. \MR{3381482}

\bibitem[BG13]{BG:book1}
Michael Baake and Uwe Grimm, \emph{Aperiodic order. {V}ol. 1}, Encyclopedia of
  Mathematics and its Applications, vol. 149, Cambridge University Press,
  Cambridge, 2013, A mathematical invitation, With a foreword by Roger Penrose.
  \MR{3136260}

\bibitem[BK04]{BK:FK}
O.~M. Braun and Y.~S. Kivshar, \emph{The {F}renkel-{K}ontorova model}, Texts
  and Monographs in Physics, Springer-Verlag, Berlin, 2004, Concepts, methods,
  and applications. \MR{2035039}

\bibitem[BT15]{BT:anti}
S.~V. Bolotin and D.~V. Treshch\"ev, \emph{Anti-integrable limit}, Uspekhi Mat.
  Nauk \textbf{70} (2015), no.~6(426), 3--62. \MR{3462714}

\bibitem[CdlL98]{CandelLlave}
A.~Candel and R.~de~la Llave, \emph{On the {A}ubry-{M}ather theory in
  statistical mechanics}, Comm. Math. Phys. \textbf{192} (1998), no.~3,
  649--669. \MR{1620543}

\bibitem[dlL01]{llave:tutorial}
Rafael de~la Llave, \emph{A tutorial on {KAM} theory}, Smooth ergodic theory
  and its applications ({S}eattle, {WA}, 1999), Proc. Sympos. Pure Math.,
  vol.~69, Amer. Math. Soc., Providence, RI, 2001, pp.~175--292. \MR{1858536}

\bibitem[dlLV10]{LlV:ground}
Rafael de~la Llave and Enrico Valdinoci, \emph{Ground states and critical
  points for {A}ubry-{M}ather theory in statistical mechanics}, J. Nonlinear
  Sci. \textbf{20} (2010), no.~2, 153--218. \MR{2608835}

\bibitem[Dwo93]{dworkin}
Steven Dworkin, \emph{Spectral theory and x-ray diffraction}, J. Math. Phys.
  \textbf{34} (1993), no.~7, 2965--2967. \MR{1224190 (94g:82049)}

\bibitem[FK39]{FK}
J.~Frenkel and T.~Kontorova, \emph{On the theory of plastic deformation and
  twinning}, Acad. Sci. U.S.S.R. J. Phys. \textbf{1} (1939), 137--149.
  \MR{0001169}

\bibitem[GGP06]{GGP:minimalFK}
Jean-Marc Gambaudo, Pierre Guiraud, and Samuel Petite, \emph{Minimal
  configurations for the {F}renkel-{K}ontorova model on a quasicrystal}, Comm.
  Math. Phys. \textbf{265} (2006), no.~1, 165--188. \MR{2217301}

\bibitem[Gol01]{gole}
Christophe Gol{\'e}, \emph{Symplectic twist maps}, Advanced Series in Nonlinear
  Dynamics, vol.~18, World Scientific Publishing Co. Inc., River Edge, NJ,
  2001, Global variational techniques. \MR{MR1992005 (2004f:37070)}

\bibitem[GPT17]{GPT:calibrated}
Eduardo Garibaldi, Samuel Petite, and Philippe Thieullen, \emph{Calibrated
  configurations for {F}renkel-{K}ontorova type models in almost periodic
  environments}, Ann. Henri Poincar\'e \textbf{18} (2017), no.~9, 2905--2943.
  \MR{3685979}

\bibitem[KP06]{Kellendonk-Putnam:RS}
Johannes Kellendonk and Ian~F. Putnam, \emph{The {R}uelle-{S}ullivan map for
  actions of {$\Bbb R^n$}}, Math. Ann. \textbf{334} (2006), no.~3, 693--711.
  \MR{2207880 (2007e:57027)}

\bibitem[Lag99]{Lagarias:Geo1}
J.~C. Lagarias, \emph{Geometric models for quasicrystals {I}. {D}elone sets of
  finite type}, Discrete Comput. Geom. \textbf{21} (1999), no.~2, 161--191.
  \MR{1668082}

\bibitem[Lif07]{lifshitz:What}
Ron Lifshitz, \emph{What is a crystal?}, Z. Kristallogr \textbf{222} (2007),
  no.~6, 313--317.

\bibitem[MF94]{forni-mather}
John~N. Mather and Giovanni Forni, \emph{Action minimizing orbits in
  {H}amiltonian systems}, Transition to chaos in classical and quantum
  mechanics ({M}ontecatini {T}erme, 1991), Lecture Notes in Math., vol. 1589,
  Springer, Berlin, 1994, pp.~92--186. \MR{1323222}

\bibitem[SBGC84]{schechtman}
D.~Shechtman, I.~Blech, D.~Gratias, and J.~W. Cahn, \emph{Metallic phase with
  long-range orientational order and no translational symmetry}, Phys. Rev.
  Lett. \textbf{53} (1984), 1951--1953.

\bibitem[ST18]{ST:SS}
Scott Schmieding and Rodrigo Trevi{\~{n}}o, \emph{Self affine delone sets and
  deviation phenomena}, Communications in Mathematical Physics \textbf{357}
  (2018), no.~3, 1071--1112.

\bibitem[ST19]{ST:random}
Scott Schmieding and Rodrigo {Trevi\~no}, \emph{{Random Substitution Tilings
  and Deviation Phenomena}}, arXiv e-prints (2019), arXiv:1902.08996.

\bibitem[{Tre}19]{T:TTT}
Rodrigo {Trevi{\~n}o}, \emph{{Tilings, traces and triangles}}, arXiv e-prints
  (2019), arXiv:1906.00466.

\end{thebibliography}

\end{document}